\documentclass[runningheads]{llncs}
\usepackage[T1]{fontenc}
\usepackage{graphicx}
\usepackage{amsmath}
\usepackage{amsfonts}
\usepackage{amssymb}

\usepackage{mathtools}
\usepackage{verbatim}

\usepackage[T1]{fontenc}

\usepackage[ruled]{algorithm}
\usepackage[noend]{algpseudocode}

\usepackage{enumerate}

\usepackage{amsmath}
\usepackage{amssymb}
\usepackage{amsfonts}
\usepackage{paralist}
\usepackage[ruled]{algorithm}
\usepackage[noend]{algpseudocode}
\usepackage{hyperref}
\usepackage{ dsfont }

\usepackage{graphicx}
\usepackage{caption}
\usepackage{subcaption}

\usepackage{epsfig}
\usepackage{enumerate}

\usepackage{chngcntr}
\usepackage{apptools}

\newcommand{\remove}[1]{}
\newcommand{\newemptyset}{\varnothing}

\usepackage{color}

\remove{
\newtheorem{definition}{Definition}
\newtheorem{theorem}{Theorem}
\newtheorem{corollary}{Corollary}
\newtheorem{lemma}{Lemma}

}

\newtheorem{observation}{Observation}

\newif\ifrobocza
\roboczafalse

\newif\ifbibtex
\bibtextrue

\ifrobocza
\newcommand{\tj}[1]{{\color{blue}{#1}}}
\else
\newcommand{\tj}[1]{#1}
\fi

\newcommand{\cL}{{\mathcal L}}

\newcommand{\eps}{\varepsilon}

\newcommand{\parent}{\textit{p}}
\newcommand{\partwo}{\text{par}}

\newcommand{\commentt}[1]{}

\begin{document}
\def\thefootnote{\fnsymbol{footnote}}

\title{Optimal-Length Labeling Schemes and Fast Algorithms for $k$-gathering and $k$-broadcasting}
\titlerunning{Labeling Schemes for $k$-gathering and $k$-broadcasting}
%
\author{Adam Ga\'{n}czorz\orcidID{0000-0001-9656-1643} \and Tomasz Jurdzi\'{n}ski\orcidID{0000-0003-1908-9458}}
\authorrunning{A. Ganczorz et al.}
%
\institute{Institute of Computer Science, University of Wroc{\l}aw, Poland.  
\email{adam.ganczorz,tju@cs.uni.wroc.pl}}
\maketitle              

\remove{
\author{Adam Ga\'{n}czorz\footnotemark[1]
	\and Tomasz Jurdzi\'{n}ski\footnotemark[1] 
}

\footnotetext[1]{Institute of Computer Science, University of Wroc{\l}aw, Poland. 
	Emails: {\tt \{adam.ganczorz,tju\}@cs.uni.wroc.pl}. 
	}
}



\newif\iffull
\fullfalse

\begin{abstract}
We consider basic communication tasks in arbitrary radio networks: $k$-broadcasting and $k$-gathering. In the case of $k$-broadcasting messages from $k$ sources have to get to all nodes in the network. 
The goal of $k$-gathering is to collect messages from $k$ source nodes in a designated sink node. We consider these problems in the framework of distributed algorithms with advice. 

Krisko and Miller showed in 2021 that the optimal size of advice for $k$-broadcasting is $\Theta(\min(\log \Delta,$ $ \log k))$, where $\Delta$ is equal to the maximum degree of a vertex of the input communication graph. We show that the same bound $\Theta(\min(\log \Delta, \log k))$ on the size of optimal labeling scheme holds also for the $k$-gathering problems. Moreover, we design fast algorithms for both problems with asymptotically optimal size of advice. For $k$-gathering our algorithm works in at most $D+k$ rounds, where $D$ is the diameter of the communication graph. This time bound is optimal even for centralized algorithms. We apply the $k$-gathering algorithm for $k$-broadcasting to achieve an algorithm working in time $O(D+\log^2 n+k)$ rounds. 
We also exhibit a logarithmic time complexity gap between distributed algorithms with advice of optimal size and distributed algorithms with distinct arbitrary labels.




    
\end{abstract}


\renewcommand{\thefootnote}{\arabic{footnote}}

\section{Introduction}
\label{s:intro}

We consider two fundamental communication tasks in arbitrary radio networks called $k$-broadcasting and $k$-gathering. In the $k$-broadcasting problem messages from $k$ sources has to be delivered to all nodes in the network. The goal of $k$-gathering is to collect messages from $k$ source nodes in a fixed sink node. 

We consider the distributed computing model of radio networks, defined as undirected connected graphs. In the absence of labels, most communication tasks including the problems considered in this paper are infeasible for deterministic algorithms in many networks, due to interferences. Hence we assume that nodes are assigned labels that are, not necessarily different, binary strings. At the beginning of an execution of an algorithm each node knows only its own label and can use it as a parameter in a deterministic algorithm executed simultaneously by all nodes. The size of a labeling scheme is equal to the largest length of a label. Another complexity measure is time understood as the largest number of communication rounds of an execution of an algorithm for given network parameters. 

\subsection{The model and the problems}

We consider radio networks (RN), where a network is modeled as a simple undirected connected graph $G=(V,E)$ with $n=|V|$ nodes. We also distinguish other parameters of the network, $D$ which is equal to the diameter of $G$, and $\Delta$ which is its maximum degree. 
With a small abuse of notation, we sometimes use $D$ to denote the height of a BFS spanning tree of a graph with a fixed root node. Observe however that the height of a BFS tree belongs to the range $[D/2, D]$, so it is of the same asymptotic size as $D$. 

We use square brackets to denote sets of consecutive integers: $[i,j] = \{i, \dots, j\}$ and $[i] = [1, i]$. 
All logarithms are to the base 2.
\remove{For clarity of presentation, we assume
that the number of nodes of a graph $n$ is a power of 2. One can easily generalize all the results for arbitrary $n$, preserving asymptotic efficiency.} 

As usually assumed in the algorithmic literature on radio networks, nodes communicate in synchronous rounds also called steps. 
In each round, a node can either transmit a message to all its neighbors, or stay silent and listen. At the receiving end, a node $v$ hears the message from a neighbor $w$ in a given round, if $v$ listens in this round, and if $w$ is its only neighbor that transmits in the given round. 
\tj{Otherwise, $v$ does not receive any message or information, and, in particular, cannot determine whether or not any of its neighbors transmitted.}
%
The \emph{time} of an algorithm for a given problem is the largest number of rounds it takes to solve it, expressed as a function of the task and various network parameters. 

If nodes are indistinguishable, i.e., in the absence of labels, the problems considered in the paper cannot be solved deterministically. Even $1$-broadcasting and $1$-gathering cannot be solved in the four-cycle due to the symmetry breaking problem. For example, for $1$-broadcasting, the node antipodal to the source cannot get its message because, in any round, either both its neighbors transmit or both are silent. 
Therefore we consider labeled networks, i.e., we assign \emph{labels} to nodes. A {\em labeling scheme} is a function $\cL$ which, for a given network represented by a graph $G=(V,E)$, assigns binary strings to the nodes of the communication graph $G$. 
The string $\cL(v)$ is called the \emph{label} of the node $v$.
Note that labels assigned by a labeling scheme are not necessarily distinct. The \emph{length} or \emph{size} of a labeling scheme $\cL$ is equal to the maximum length of any label assigned by it. Every node knows initially only its label, and can use it as a parameter in the deterministic algorithm executed simultaneously and distributively by all nodes.



Solving distributed network problems with short labels can be seen in the framework of
algorithms with \emph{advice} (i.e., with \emph{labeling schemes}).  
\tj{In this setting, nodes are anonymous and do not know anything about the network. An oracle that knows the network provides advice, i.e., a binary string, to each node before the beginning of the computation.}
A distributed algorithm executed at nodes uses these strings (each node uses its own advice string) to solve the problem efficiently.
The required size of advice equal to the maximum length of the strings can be seen as a measure of the difficulty of the problem.
In the radio network literature regarding algorithms with advice different strings may be given to different nodes \cite{EllenG20,DBLP:conf/spaa/EllenGMP19,DBLP:journals/mst/FraigniaudKL10,DBLP:journals/iandc/FuscoPP16} and we also make this assumption in the present paper. 
The scenario of the same advice given to all nodes would be useless in the case of deterministic algorithms in radio networks: no deterministic communication could occur.
If advice strings may be different, they can be considered as labels assigned to nodes.
Such labeling schemes permitting to solve a given network task efficiently are also called {\em informative labeling schemes}.


We now formally define communication problems considered in this paper. 
Let $G=(V,E)$ be the communication graph of a radio network.
The \emph{$k$-{gathering}} problem assumes that there are $k$ nodes $s_1,\dots,s_k \in V$, called {\em sources}, and one node $s\in V$, called the sink. Each source has a message, and all messages have to reach the \emph{sink}.
We distinguish the unit-message $k$-broadcasting and the multi-message $k$-broadcasting.
In both cases there are $k$ nodes $s_1,\dots,s_k \in V$, called {\em sources}.
In the case of \emph{unit-message $k$-broadcasting} all the sources know the same broadcast message $M$ while in the case of \emph{multi-message $k$-broadcasting} the nodes $s_1,\ldots,s_k$ initially have  possibly distinct messages $M_1,\ldots,M_k$. The goal is to deliver the message $M$ (in the unit-message case) or  all the messages $M_1,\ldots,M_k$ (for the multi-message $k$-broadcasting) to  all nodes of the communication graph $G(V,E)$. 
The unit-message $k$-broadcasting problem with $k=1$ is just called the \emph{broadcasting problem}.

We assume the \emph{non-spontaneous wake-up} model. 
In this model only source nodes start an execution of an algorithm at the fixed round $0$  and each other node might join the execution only after receiving a message from other node. 
This
model is in general more demanding than the \emph{spontaneous wake-up} model, where all nodes start an execution of an algorithm at the same round (see e.g. \cite{DBLP:journals/jacm/CzumajD21}).
However, as all problems considered in this paper assume that each node is
either awaken already at the round $0$ of the algorithm 
or it can
only be awaken by a message from another node, we can think of the number of rounds elapsed from the round $0$ of an execution of an algorithm as the current value of a \emph{central clock}. Then, we can add the current value of the central
clock to each transmitted message and guarantee 
that each node learns the state of the central clock.

%
As usual in algorithmic literature concerning radio networks we assume
that many messages can be combined together into a single packet transmitted in a round, so a node can transmit several original messages $m_i$ in one round. Thus, there is no restriction on the size of messages sent by nodes.

The ultimate goal for the framework with advice is to design algorithms which use 
optimal size of labels and simultaneously work in optimal time.
%


\subsection{Our results}


 We present an algorithm for $k$-gathering, using a labeling scheme of length $O(\min(\log k,$ $\log \Delta))$, and running in at most $D + k$ rounds. We show that the length of our labeling scheme is optimal 
 and 
 the time 
 is also asymptotically optimal, even for centralized algorithms
 in which the 
 topology of the communication graph is \textit{known} to nodes. 
 For 
 $k$ such that $k=\omega(D\Delta)$,\footnote{As an example of such a case one may consider a graph with a spanning tree of height $O(\log \Delta)$. Then $D$ is also $O(\log \Delta)$, while $k$ might be even equal to $n$ if each node is a source of a message. And thus $O(D\Delta)=o(k)$ for wide range of values of $\Delta$.} we provide a $k$-gathering algorithm which works in time $O(D\Delta)$. Thus, we 
 accomplish 
 $k$-gathering in $O(\min(D+k,D\Delta))$ rounds.
\tj{By combining $k$-gathering and broadcasting, we get an algorithm for the multi-message $k$-broadcasting problem. 
Thanks to the $O(D+\log^2 n)$-round algorithm for broadcasting with $O(1)$ labels \cite{DBLP:journals/corr/abs-2410-07382} and $O(D+k)$-round $k$-gathering with advice of size $O(\min(\log k,\log \Delta))$ provided here, we get a $O(D+\log^2 n+k)$ time algorithm for multi-message $k$-broadcasting with labels of optimal size $O(\min(\log k,\log \Delta))$, where the optimality of the size of labels follows from \cite{DBLP:conf/wg/KriskoM21}. As for the unit-message broadcasting problem, our reduction to the standard broadcasting problem gives the algorithm with $O(1)$ length of labels and $O(D+ \log^2 n)$ rounds.}
	
 Thus, our solutions for the considered communication tasks use labels of asymptotically optimal length. Moreover  distributed algorithms for $k$-gathering using those labels work in optimal time, even for algorithms working in the scenario with known topology of the input graph (to each node).

\subsection{Related work}

The most studied related tasks 
are broadcasting and gossiping. Here gossiping is the problem in which each node has its own message and all messages have to be distributed to all nodes.
%
Optimal time centralized broadcasting algorithms were given in \cite{DBLP:journals/dc/GasieniecPX07,KP-DC-07} and the best known gossiping time (without any extra assumptions on parameters) follows from \cite{DBLP:journals/dc/GasieniecPX07}. 
They also consider gossiping for which they design an algorithm which works in $O(D+\Delta\log n)$ rounds.
For large values of $\Delta$, the gossiping from \cite{DBLP:journals/dc/GasieniecPX07} was later improved in \cite{CMX}.  
The best known deterministic distributed algorithm for broadcasting which time that depends only on $n$ is $O(n\log n\log\log n)$ \cite{DeMarco}, later improved in \cite{CD} for some values of parameters $D$ and $\Delta$. For gossiping, the best known time  in directed strongly connected graphs was given in \cite{GL,GRX} and 
for undirected graphs 
in \cite{Vaya}. 
The 
$n$-gathering 
problem in asymmetric directed graphs \cite{DBLP:journals/iandc/ChrobakCG21} and trees was considered in \cite{DBLP:journals/algorithmica/ChrobakC18,DBLP:journals/iandc/ChrobakCGK18}.


The paradigm of algorithms with advice has been applied to many different distributed network tasks: finding a minimum spanning tree  \cite{DBLP:journals/mst/FraigniaudKL10}, finding the topology of the network \cite{DBLP:journals/iandc/FuscoPP16,DBLP:journals/iandc/GanczorzJLP23,DBLP:journals/tcs/GorainP21}, 
\tj{determining the size of a network \cite{DBLP:journals/iandc/GanczorzJLP23,DBLP:journals/tcs/GorainP21a}, single-message broadcasting \cite{DBLP:journals/tcs/IlcinkasKP10},}
and leader election \cite{DBLP:journals/talg/GlacetMP17}. In \tj{\cite{DBLP:conf/netys/BuPR20,EllenG20,DBLP:conf/spaa/EllenGMP19,DBLP:journals/corr/abs-2410-07382},} the broadcasting algorithm with advice of size $O(1)$ were provided working in time $O(n)$, $O(D\log n+\log^2n)$ and optimal $O(D+\log^2n)$, respectively. 
%

Krisko and Miller \cite{DBLP:conf/wg/KriskoM21} were the first who considered the $k$-broadcasting problem for radio networks with advice. They show that $\Theta(\min(\log k,\log \Delta))$ is the optimal size of labels and they design a $k$-broadcasting algorithm for general graphs and for trees without optimizing time of the algorithms. \tj{Bu et al.\ \cite{DBLP:journals/tcs/BuLPR23} introduced the study of the labeling for the the convergast problem which corresponds to the $n$-gathering problem. They however mainly consider radio networks with collision detection, a model which is stronger than the one considered in this paper. They also provide a solution which works in our model in $\Theta(n)$ time and $\Theta(\log n)$-bit labels, which matches the extreme (and simple) case of $k$-gathering with $k=n$.}

\subsection{Organization of the paper}
In Section~\ref{sec:k:gathering} 
we present 
$k$-gathering algorithms. 
Section~\ref{sec:k:broadcast} 
regards the $k$-broadcasting problem.
We discuss further research directions in Section~\ref{sec:summary}.

\section{ The $k$-gathering problem}
\label{sec:k:gathering}

In this section we consider 
$k$-gathering. 
Section~\ref{subsec:kgather:lower} gives a lower bound $k+D$ on the time needed 
for $k$-gathering and a lower bound
$\min\left( \log k, \log \Delta \right)$ on the length of a labeling scheme.
In Section~\ref{subsec:gathering:tools}, we introduce some technical tools 
used in a $k$-gathering algorithm that works in time $O(D+k)$, presented in Section~\ref{subsec:k:gathering}. 
In Section \ref{subsec:k:gather:big:k} we present an algorithm which works in $O(D\Delta)$ rounds. 
%
Finally, in Section~\ref{subsec:k:gather:sum}, we summarize 
the results 
and show a logarithmic gap in time complexity between algorithms with and without advice.

\subsection{Lower bounds for the $k$-gathering problem}\label{subsec:kgather:lower}
In this section we show that $D+k$ is a lower bound on the number of rounds needed to accomplish $k$-gathering, and that $\min\left( \log\Delta,\log k\right)$ is a lower bound on the length of a labeling scheme sufficient to accomplish this task. 

Consider the graph $G_{D,p}$ with 
nodes $V=\{v_1,v_2,\ldots, v_{D+p}\}$,
and 
the 
edges
\newcommand{\figkgatheringlower}{
\begin{figure}[h]
        \centering
        \includegraphics[width=0.85\linewidth]{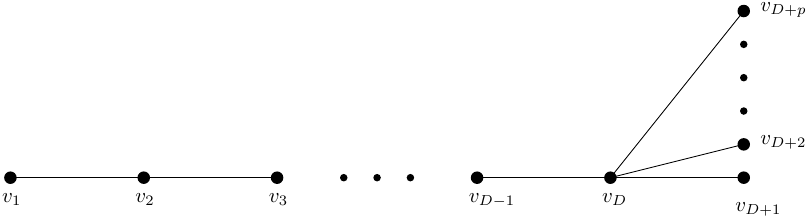}
        \caption{Illustration for the lower bounds on $k$-gathering.}
        \label{fig:kgathering-lower}
\end{figure}    
}
\iffull
$\{(v_i,v_{i+1})\,|\, i\in[1,D]\}\cup \{(v_{D}, v_j)\,|\, j\in[D+1,D+p]\}$ 
(see Figure~\ref{fig:kgathering-lower}).
\else
$\{(v_i,v_{i+1})\,|\, i\in[1,D-1]\}\cup \{(v_{D}, v_j)\,|\, j\in[D+1,D+p]\}$ (see Appendix~\ref{sec:fig:kgathering:lower}).
\iffull
\figkgatheringlower
\fi

Assume that $p=k$, $k$ input messages are originally located in the nodes $v_{D+1},\ldots,v_{D+k}$
\iffull{(see Figure~\ref{fig:kgathering-lower})}\fi and the messages should be gathered in $v_1$. As the unique path to $v_1$ of each message goes through $v_D$, we see that at least $k$ rounds are needed in order to transfer all  messages from their source nodes to $v_D$. Then, an additional $D-1$ rounds are necessary to pass any message from $v_D$ to $v_1$. These observations imply the lower bound $D+k$ on the number of rounds for 
$k$-gathering.

For the lower bound on the size of labels, look again at the 
\iffull
graph from Figure~\ref{fig:kgathering-lower}.
\else
graph.
\fi
Firstly assume that $p=k$, the source messages are originally located in the nodes $v_{D+1},\ldots,v_{D+k}$ and they should be gathered at $v_1$.
Then we need the labels of size $\log k$ in order to guarantee that all nodes $v_{D+1},\ldots,v_{D+k}$ have distinct labels necessary for delivery of the source messages to $v_1$. 
Indeed, if two nodes $v_{D+i}$ and $v_{D+j}$ for $i\neq j$, $i,j\in[k]$ have equal labels and they are connected only to $v_D$, then both of them either transmit or listen in each round. Thus, none of them transmits its message successfully to its only neighbor $v_D$ and therefore the given instance of $k$-gathering problem can not be accomplished.
Secondly, if $p=\Delta<k$ and $\Delta$ of the $k$ source messages are located originally in $v_{D+1},\ldots,v_{D+\Delta}$, the labels of size $\log \Delta$ are necessary.
            

The above considerations lead to the following conclusion. 

\begin{theorem}\label{th:kgathering:lower}
    The $k$-gathering task requires at least $k+D$ rounds and a labeling scheme of size
    $\min\left( \log k, \log \Delta \right)$.
\end{theorem}
We note that the above lower bound on the size of labels could be also obtained by a reduction of the $k$-broadcasting problem from
\cite{DBLP:conf/wg/KriskoM21} to a composition of a $k$-gathering and standard broadcasting and the fact that broadcasting can be 
done with $O(1)$-bit labels.
The lower bound on $k$-broadcasting is given 
in \cite{DBLP:conf/wg/KriskoM21}.

\subsection{Tools for $k$-gathering algorithms}\label{subsec:gathering:tools}
\begin{definition}[Children-parents assignment]\label{def:cpa}
Let $G=(V,E)$ be a 
graph and let:\\
--~$X=\{x_1,\ldots, x_p\}$ such that $x_i\in V$ and $x_i\neq x_j$ for each $i\neq j$ be a set called \emph{children}, \\
%
--~$Y=\{y_1,\ldots,y_r\}\subseteq$ $V\setminus \{x_1,\ldots,x_p\}$ such that $(x_i,y_j)\in E$ for each $i\in[p]$ and 
\tj{at least one $j\in[r]$}
be a set called \emph{parents},\\ 
%
--~$\partwo\,:\, X\to Y$ such that for each $y\in Y$, there exists $x\in X$ 
    such that $(x,y)\in E$ and 
    $\partwo(x)=y$ \tj{(i.e., $\partwo$ is a surjective function).\\}
Then, we say that $(X,Y, \partwo)$ is a \emph{children-parents assignment}, $x$ is a \emph{child} of $y$ and $y$ is the \emph{parent} of $x$ iff $\partwo(x)=y$.
The function $\partwo$ is called a \emph{parent} function.
\end{definition}
{In the following part of the paper we use the notion of a \emph{schedule} in order to describe some distributed algorithms for radio networks with advice. Formally, a \emph{schedule} $S$ might be seen as an assignment of all possible labels of nodes to 0/1 sequences of a fixed size $T$, where $S_\lambda$ denotes the 0/1 sequence assigned to the label $\lambda$. Such the schedule $S$ is understood as an algorithm in radio networks, where each node with the label $\lambda$ transmits in the round $i$ of the algorithm corresponding to $S$ iff the value of the $i$th bit of $S_\lambda$ is equal to $1$. In general, the assumptions of the non-spontaneous wake-up model might prevent a possibility of an execution of a schedule as a distributed algorithm. However, as we explained earlier, all problems considered in this paper assume that each node is either awaken already at the first round of the algorithm (if it is a source of a message) or it can only be awaken by a message from another node. Thus, we can add the current value of the central clock to each message and guarantee in this way that each node learns the state of this central clock before it is supposed to send its first message.}
\begin{definition}\label{def:cps}
A schedule $S$ is called an $(X,Y,\partwo)$ children-parents schedule 
in a graph $G(V,E)$ for a children-parents assignment $(X,Y,\partwo)$
iff $x$ successfully sends a message to $\partwo(x)$ for each $x\in X$ during an execution of $S$ in $G$.
\end{definition}

Our ultimate goal is to use children-parents schedules as stages of $k$-gathering algorithms guaranteeing that each message gets closer to the sink node for the $k$-gathering task. Therefore, we would like to make progress towards finalizing $k$-gathering in a children-parents schedule either by minimizing the size of the schedule or by decreasing dispersion of messages measured by the number of nodes in which the messages to be gathered are located (here, decreasing dispersion corresponds to decreasing the number of parent nodes).
To this aim, 
we will show that there is a way to 
rearrange a children-parents assignment such that the sum of the sizes of the optimal children-parents schedule and the number of 
\tj{target}
nodes $|Y|$ is bounded from above.

\begin{lemma}\label{l:level}
	Let $(X,Y,\partwo)$ be a children-parents assignment in a graph $G(V,E)$. Then, there exists
	$Y_\star\subseteq Y$, $\partwo_{\star}\, : \,X\to Y_\star$ and a schedule $S$ such that \\
        --~$(X,Y_\star, \partwo_\star)$ is a children-parents assignment in $G$,\\
        --~$S$ is a $(X,Y_\star,\partwo_\star)$ children-parents schedule, \tj{such that each node from $X$ transmits exactly once in $S$,}\\
        --~$|Y_\star| \leq |X|- |S| +1$, i.e., $|Y_\star| + |S|\leq |X| +1$.
\end{lemma}

\begin{proof}

We will proceed by induction on the size of $X$. 
 
Consider the base case that $|X|=0$. Then $S$ is an empty schedule. Thus $Y_\star=Y=\newemptyset$,   $|Y_\star|=|Y|=0$, and therefore $|Y_\star|\le |X|$ by Def.~\ref{def:cpa} and $|S|=0$. These inequalities imply that    $|Y_\star|\le |X|-|S|+1$ as required.
	
Now assume that the lemma holds for all children-parents assignments such that the number of children is smaller than $k>0$.
 Let us take any children-parents assignment $(X,Y,\partwo)$ and let $|X|=k$. Let us look at the subgraph of $G$ induced on $X\cup Y$. 
Let $y_k \in Y$ be a node with some degree $d>0$, i.e., $\deg(y_k) = d$. Let $X' = \{x \in X| (x, y_k) \notin E\}$ and $Y' = \{y | \exists x \in X'\  par(x) = y\}$ (see Figure~\ref{fig:ch:par:schedule}).
        \begin{figure}
        \centering
        \includegraphics[width=0.75\linewidth]{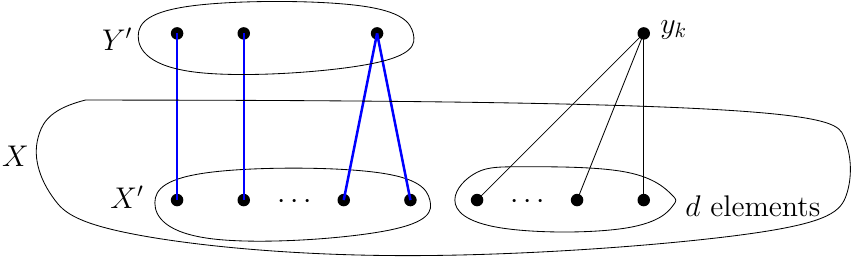}
        \caption{Inductive step in the proof of  Lemma~\ref{l:level}, Case~1. Segments denote edges, fat blue segments correspond to edges connecting elements of $X'$ with their parents. There are no edges connecting nodes from $X'$ with $y_k$ in $G(V,E)$. }
        \label{fig:ch:par:schedule}
        \end{figure}        
 Let $\partwo'$ be the function $\partwo$ with the domain reduced to $X'$. Then $(X',Y',\partwo')$ is a correct children-parents assignment with $|X'| = |X| -d$,$|Y'| \le |Y|-1$. From the inductive hypothesis we get a children-parents assignment $(X',Y'_\star,\partwo'_\star)$ for $Y'_\star \subseteq Y'$ and a $(X',Y'_\star,\partwo'_\star)$ children-parents schedule $S'$ such that  $|Y'_\star| \leq |X'|-|S'|+1$.
	
\noindent Now let us look at two cases:
    
    \noindent \textbf{Case 1:} There is a child $x$ of $y_k$ such that, for each node $y\in Y'_\star$, there is no edge $(x, y)$ in $G$.\\ Then the node $x$ does not collide with any node in $X'$ at any element of $Y'_\star\cup\{y_k\}$ as all nodes of $X'$ are not connected with $y_k$ and $x$ is not connected with nodes of $Y'_\star$ according to our assumption. Let us modify $S'$ in the following way.
        Firstly, we send a message of $x$ in the first round of the schedule $S'$ in parallel of other transmission of $S'$ of that round. According to our above assumption and reasoning, transmission of $x$ is successfully delivered to $y_k$ and it does not interfere (does not causes collisions) with any transmission of other nodes in the first round of $S'$.
        Secondly, after $S'$, we assign a separate round for each of the remaining children of $y_k$ in order to allow them to transmit without collisions. That gives us a schedule $S$ of length $|S| \leq |S'| + d-1$ as $y_k$ has $d>0$ children. 
        Moreover, we obtain the new set of parents $Y_\star=Y'_\star\cup\{y_k\}$ 
        and therefore 
        \tj{$|Y_\star| = |Y'_\star| + 1 \leq |X'| - |S'| + 1 + 1 \le \left(|X| - d\right) - (|S| - d + 1) + 2 = |X| - |S| + 1,$}
        where the first equality follows from the construction of $Y'_\star$, the first inequality follows from the inductive hypothesis, the second equality follows from the relationships $|X'|=|X|-d$ and $|S'|=|S|-d+1$ which in turn follow from our construction of $X'$ and $S'$.

\noindent \textbf{Case 2:} Each child  $x_i$ of $y_k$  has an edge to some $y'_i \in Y'_\star$ for all nodes $i \in [d]$.\\ Then we will assign $y'_i$ as the (possibly different than according to the function par) parent $\text{par}'(x_i)$ of $x_i$ for each $i\in[d]$. This enables us to avoid using $y_k$ or any other node outside of $Y'_\star$ as a parent of any node. Our schedule $S$ will be then just $S'$ followed by $d$ rounds such that the $i$th neighbor of $y_k$ is the only transmitter in the round $i$. Then the size of $S$ is $|S|=|S'|+d$ and $|Y_\star| = |Y'_\star| \leq |X'| -|S'| + 1 = |X| - d - (|S| - d)  +1 = |X|-|S| +1,$
        where the inequality follows from the inductive assumption and the equalities follow from the construction. In this way we finish the inductive proof of the lemma.
\qed
\end{proof}

\subsection{The $k$-gathering in $D+k$ rounds}\label{subsec:k:gathering}

In this section 
we design an algorithm which accomplishes $k$-gathering
 in $D+k$ rounds with labels of size $O\left(\min(\log\Delta,\log k)\right)$.
The idea of the algorithm is as follows. 
Let $T$ be a BFS tree rooted at the sink node of a given instance of the $k$-gathering problem. The algorithm will work in phases. In the first phase we will build a children-parents assignment $A$, along with a children-parents schedule $S$ on leaves of $T$ (containing a message) as children, and their parents
in $T$. 
\tj{Thus, after an execution of $S$, the parent nodes from $A$ will get all messages from the children. 
More precisely, we will actually execute a schedule $S'$ based on $S$ that ensures that, for each leaf $v$ in $T$, each message located in $v$ is received by at least one parent of a leaf in $T$ (not necessarily by the parent of $v$ as assigned by the par function of the children-parent assignment $A$).}

After the first phase, we start the second phase which now is a children-parents schedule with children equal to the leaves of $T'$, where $T'$ is a subtree of $T$ obtained by removing
the leaves of $T$.
The tricky parts 
are in the encodings of behaviors of nodes in the schedules using the labels of size $O(\min(\log \Delta, \log k))$.
Other non-obvious observations are needed in the evaluation of time of the composition of the 
children-parents schedules, which requires amortized analysis. 

First, we present a centralized implementation of the above 
strategy. Then, 
an optimal labeling scheme is presented.
Finally, 
a distributed algorithm simulating the centralized one and using assigned labels is described and analyzed.

\subsubsection{Centralized transmission scheme for $k$-gathering in $D+k$ rounds}
\label{s:gather:dplusk:offline}

In this section we describe more precisely the centralized algorithm, which will serve as a base for a distributed algorithm.

The algorithm works in phases.
In the $i$th phase, we will create a children-parents schedule $S_i$. Let $d(v)$ for each node $v$ be the distance from $v$ to the sink node $s$.
Let $T_0$ be a tree consisting of shortest paths from source nodes to the sink node $s$. 
Then, let $X_0$ be the set of leaves of $T_0$, let $Y_0$ be the set of parents of the nodes from $X_0$
and $\partwo_0$ be the corresponding parent function. Let $S_0$, $Y'_0$, $\partwo'_0$ be the schedule, the parent set, and the parent function obtained by an application of Lemma~\ref{l:level} for the children-parents assignment $(X_0,Y_0,\partwo_0)$. Moreover, let $T_1$ be a subtree of $T_0$ obtained by removing the set $X_0$ of leaves of $T_0$ -- see Figure~\ref{fig:par-schedule-stage} in App.~\ref{sec:app:figure} for an example illustration.
If some nodes become leaves without any message after this removal, we remove them from the tree. 
\newcommand{\figparschedulestage}{
        \begin{figure}[h]
        \centering
        \includegraphics[width=0.85\linewidth]{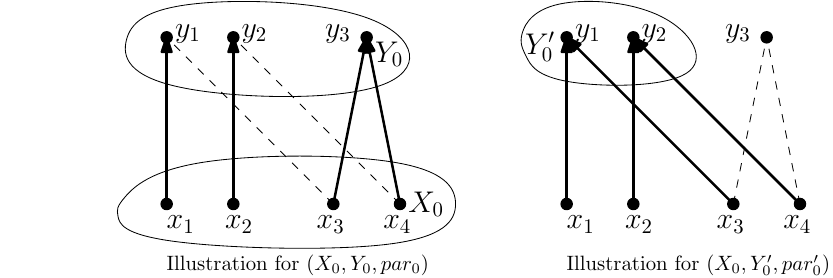}
        \caption{An example for $(X_0,Y_0,\partwo_0)$ and $(X_0,Y'_0,\partwo'_0)$. Fat edges with arrows connect children with their parents. 
        Dashed edges do not connect parents and children. The schedule $S_0$ might be such that $x_1$ and $x_2$ transmit simultaneously in the first round, $x_3$ and $x_4$ transmit simultaneously in the second round. \tj{(Note that the parents of $x_3$ and $x_4$ in $\partwo'_0$ are $y_1$ and $y_2$, respectively.)} If messages are originally located only in $x_1,\ldots,x_4$ then $x_1,\ldots,x_4$ as well as $y_3$ are removed from $T_0$ in order to obtain $T_1$.}
        \label{fig:par-schedule-stage}
        \end{figure}        
}

For $i>0$, let $Z_i$ be the set of source nodes of $T_i$ outside of $Y'_{i-1}$.
\tj{Then, let $\text{Leaves}(T_i)$ be the set of leaves in $T_i$, let   $X_i=\text{Leaves}(T_i)\cap\left(Y'_{i-1}\cup Z_i\right)$,}
 let $Y_i$ be the set of parents of $X_i$ \tj{in $T_i$} and $\partwo_i$ be the corresponding parent function.  Then let $S_i$, $Y'_i$, $\partwo'_i$ be the schedule, the parent set, and the parent function respectively obtained by an application of Lemma~\ref{l:level} for $(X_i,Y_i,\partwo_i)$. \tj{Then remove $X_i$ from $T_i$ to form $T_{i+1}$, and also remove nodes that become leaves without any message after this removal.}

The resulting centralized scheme for $k$-gathering  $S$ would be thus a concatenation of 
\tj{schedules}
$S_0, S_1,\dots, S_{m-1}$. Below we prove that such $S$ is actually a $k$-gathering schedule and that its time is at most $D+k$.

\begin{lemma}
    The broadcast schedule $S$ is a correct schedule for $k$-gathering. The length of $S$ is at most $D+k$.
\end{lemma}
\begin{proof}
    Let us analyze the path of some message $a$, located originally at a node $u_0$ of $T_0$, to the sink node $s$.
    \tj{The node $u_0$ eventually becomes a leaf in $T_i$ for some $i$, so $u_0\in X_i$ in phase $i$, and then the schedule $S_i$ obtained from Lemma~\ref{l:level} guarantees that $u_0$ will successfully transmit its message to a parent $u_1$ with $d(u_1)=d(u_0)-1$.}
    We can continue this reasoning with each next $u_l$ until the message of $u_0$ reaches the node $u_{d(u_0)}=s$.
    Thus, each message gets delivered to the source node.

    Now, let us analyze the length of $S$. 
    First observe that the number of phases is equal to the largest distance of a leaf of $T_0$ to $s$ which is $m\leq D$.
    %
    By Lemma~\ref{l:level}, the length of $S_i$ is at most $|X_i|-|Y'_{i}|+1$ for $i\ge 0$. 
    As $m \leq D$, $|X_0|+\sum_{i=1}^{m-1} |Z_i| \tj{\le} k$,\footnote{Recall the assumption that each pair of original messages are located in different nodes at the beginning.} 
    $|X_i| = |Y'_{i-1}|+|Z_i|$ 
    for $i\ge 1$, and $Y'_{m-1} = \{s\}$,
    we get the following bound on the size of $S$:
    $
    |S| = \sum_{i=0}^{m-1} |S_i| \leq \sum_{i=0}^{m-1} \left(|X_i|-|Y'_{i}|+1\right)  = 
    |X_0| + \sum_{i=1}^{m-1} |Z_i| - |Y'_{m-1}| + m \leq k + D - 1.$ \qed
\end{proof}

\subsubsection{Labeling scheme for $k$-gathering in $D+k$ rounds}
\label{s:gather:dplusk:distributed}

Now we will describe the labeling scheme for a distributed implementation of the above $O(D+k)$-round schedule for $k$-gathering. Then, in the next section, a $k$-gathering algorithm using this labeling scheme will be presented. 

Firstly, let us assume that $k \leq \Delta$. In this case we will encode behaviors of nodes in the centralized schedule presented earlier, using $O(\log k)$-bit labels.
Let $v$ be a node which sends its message in the $i$th phase, that is during a schedule $S_i$. 
Note that each node transmits only once during the above described centralized schedule. \tj{(This property follows from two facts. Firstly, each node with a message to send is a leaf in $T_i$ for exactly one value of $i$. Secondly,  a node is only included in $X_i$ when it is a leaf. So each node with a message to send transmits during exactly one schedule $S_i$. Moreover, by Lemma~\ref{l:level}, each node in $X_i$ transmits exactly once in $S_i$.)}
Let $t_v$ be the round number during $S_i$ when $v$ transmits a message and $T_i$ be the length of $S_i$. 
Let us assign unique identifiers from the set $\{1,\ldots,k\}$ to all original messages which are supposed to be gathered.
Let $\texttt{ID}_v$ be the identifier of the message, which $v$ gets at the latest, during the schedule $S_{i-1}$, according to the centralized algorithm or  $\texttt{ID}_v$ is the identifier of the message originally located in $v$ if $v$  is  is \tj{an element of $Z_i$.} 
Then the label of $v$ will be equal to $\mathcal{L}(v) = \left(0, \texttt{ID}_v, t_v, T_i\right)$, where 
the value $0$ at the first coordinate indicates that we are in the case $k \leq \Delta$. 
By Lemma~\ref{l:level} and the fact that there are only $k$ messages, the length of each schedule $S_i$ is $O(k)$ and thus the length of $\mathcal{L}(v)$ is $O(\log k)$.

Now, 
consider the case that $k > \Delta$. Firstly, let 
$c: V\to $ $\{0,1,\ldots,\Delta^2\}$ be a
$2$-hop $(\Delta^2+1)$-coloring  of the input graph. 
That is, $c(v)$ is the color of $v$. 
%
%
Let $v$ be a node which sends its message in the $i$th phase.
As $k>\Delta$, the length $T_i$ of $S_i$ in our centralized schedule might be much larger than $\Delta$ as well as the round $t_v$ of the transmission of a node during $S_i$. 
Therefore, to encode $T_i$ and $t_v$ using $O(\min(\log k,\log \Delta))=$ $O(\log \Delta)$ bits,
we will optimize the schedule $S_i$ in such way that its length is always $O(\Delta^2)$, using the following observation.
%
\begin{observation}\label{obs:max:len:sched}
    Each children-parents schedule can be optimized to be of length at most \tj{$\Delta^2+1$.}
\end{observation}

Observation~\ref{obs:max:len:sched} is derived from \tj{$2$-hop $(\Delta^2+1)$-coloring.} 
If a children-parents schedule is longer than $\Delta^2+1$, we can exchange it with the round-robin protocol on the values $c(v)$ of the $2$-hop $(\Delta^2+1)$-coloring. Given a $2$-hop coloring, nodes with the same color can transmit simultaneously such that all their neighbors receive collision-free messages transmitted to them. This fact follows from the observation that, in the radio network model, messages of two distinct transmitters $u,v$ might collide only at their common neighbor which implies that the distance between $u$ and $v$ is at most $2$ in such a case.

Let $S'_i$ be such an optimized schedule of length $\min(|S_i|,\Delta^2+1)$. 
Again, let $t_v$ be the unique time step during $S'_i$ when $v$ transmits a message and $T_i$ be the length of $S'_i$. 
Let $c_v$ be the color of a node that sends its message to $v$ at the latest during $S'_{i-1}$. 
Then the labeling for $v$ is equal to $\mathcal{L}(v) = \left(1, c_v, t_v, T_i\right)$, where 
$1$ at the first coordinate indicates that we are in the case 
\tj{$k>\Delta$.}
As the length of $\mathcal{L}(v)$ is $O(\log \Delta)$ in this case,
we get the following lemma.
\begin{lemma}\label{l:gather:labellength}
    The length of the labeling scheme is $O(\min{(\log \Delta, \log k)})$.
\end{lemma}

\subsubsection{Distributed algorithm for $k$-gathering in $D+k$ rounds}\label{subsubsec:kgather:Dplusk}
Now, we will describe the distributed algorithm for $k$-gathering using labels constructed according to the labeling scheme described above. 
We 
consider two cases 
determined by the first coordinate of the label $\mathcal{L}(v)$ of each node $v\in V$.

Firstly consider the case that  $k \leq \Delta$.  
Each node from $X_i$ will, apart from all received messages from the current instance of the $k$-gathering problem, add the value of $T_0+\dots+T_{i}$ to its 
messages. Thus each node $v$ will be able to deduce the value of $T_0+\dots+T_{i-1}+t_v$ -- the round number when $v$ should transmit. Indeed, this value will be determined based on the content of received messages and the own label of the considered node $v$.
However, although the nodes from $X_i$ know $T_i$ from their labels, they do not know the sum $\sum_{j=0}^{i-1} T_i$. This value might be too large to encode it in the label of size $O(\min(\log\Delta,\log k))$. Therefore, we do not store it in the label of a node from $X_i$.
However, we show below that actually each node $v\in X_i$ will learn the value of $T_0+\cdots+T_{i-1}$ before the round of its actual unique transmission according to the presented above centralized schedule for $k$-gathering and therefore it can determine its transmission round and send $\sum_{i=0}^i T_i$ with its message to its parent.
 According to the centralized schedule, $v\in X_i$ should broadcast in the phase $i$. If $i=0$, $v$ will just broadcast in the round $t_v$ encoded in its label. If $i>0$ then there is some node $u$ relaying its messages to $v$ in phase $S_{i-1}$, along with the value of $T_0+\dots+T_{i-1}$. The node $v$ knows when it gathered all messages if it receives a message containing an ID equal to $\texttt{ID}_v$ from 
 the
 label $\mathcal{L}(v)$. Then $v$ knows that it must transmit in the next phase and $\sum_{j=0}^{i-1}T_j+t_v$ determines the exact round number of its transmission. \tj{Moreover, as $v$ knows $T_i$ from its label, it can determine $\sum_{j=0}^{i}T_j$ to be transmitted to its parent.}

We reason similarly in the case that
$k > \Delta$. Each node $v$ will, in addition to all transmitted messages, append to its broadcast messages the value of $T_0+\dots+T_{i}$. In this case however, it will also append $c(v)$. Thus the node $v$ knows it gathered all messages if it receives a message from a node with its color equal to $c_v$ from $v$'s label $\mathcal{L}(v)$. 
\tj{(Note that there can only be one such message, because the distance-two coloring guarantees that all of $v$'s neighbors have distinct colors.)}
The node $v$ thus proceeds as in the case that $k \leq \Delta$.

The above 
reasoning lead 
%
to the following theorem.
\begin{theorem}\label{t:k-gather}
	There is a $O(\min(\log k, \log \Delta))$-bit labeling scheme and a distributed algorithm for $k$-gathering using this labeling  running in at most $D+k$ rounds.
\end{theorem}
Note that the upper bound on the number of rounds 
just $D+k$, not only an asymptotic $O(D+k)$. This bound matches the lower bound from Section~\ref{subsec:kgather:lower}.

\subsection{$k$-gathering for $k=\omega(D\Delta)$}\label{subsec:k:gather:big:k}
Now, we give an algorithm for the specific case that $k$ is large, $k=\omega(D\Delta)$.

%
For a given rooted tree $T$, let $\parent_T(v)$ denote the parent of the node $v$ in $T$. If the considered tree is clear from the context, we write $\parent(v)$ instead of $\parent_T(v)$.

Lemma 6 from \cite{DBLP:journals/iandc/GanczorzJLP23}
states that there is an efficient way to assign the value $s(v) \in [0, \Delta-1]$ to each node $v$ of any BFS tree $T$ so that \\
    (a)~$s(u)\neq s(w)$ for different children $u,w$ of a node of $v$ in $T$,\\
    (b)~there is no edge between $u$ and $\parent(w)$ for any two nodes $u\neq w$ on the same level of $T$ such that $s(u)=s(w)$. \\
Now, consider a fixed BFS tree with the sink node $s$ as the root of $T$.
Then, let $T'$ be the smallest with respect to set inclusion subtree of $T$ which contains the root $s$ of $T$ and all the sources $s_1,\ldots,s_k$. Note that each leaf of $T'$ is a source vertex.
Let $s\,:V\,\to [0,\Delta-1]$ be a function which satisfies the above properties (a)--(b).
Let $c\,:\,V\to[0,\Delta^2]$ be a $(\Delta^2+1)$ distance-two coloring of the input graph.
Let the label of a node $v$ consist of the values of $s(v)$, $c(v)$, $c(\parent(v))$, $\Delta$, the number of children of $v$ in $T'$, and $l(v)\in[0,2]$ which is equal to the  level of $v$ in $T$ (i.e., its distance to the root of $T$) mod 3.

Now, we describe a $k$-gathering algorithm $A$ using the above 
labels. An execution 
is split into phases consisting of three rounds $\{0, 1, 2\}$. If 
$v$ is supposed to send a message in a given phase, then it transmits 
in the round $l(v)$ of the phase. In this way 
messages 
from different levels do not collide with each other.

At the beginning of 
the algorithm, only the leaves of $T'$ are active. A leaf 
$v$ transmits its message in the phase $s(v)$. 
Each non-leaf node $v$ (i.e., the number of its children stored in its label is larger than $0$) 
does not send any message, until messages from all its children are received. If $v$ receives its last message from its child in phase $t$, it transmits all received messages in the closest phase $t'>t$ such that $t'\mod \Delta=s(v)$ and it remains inactive in other phases. 
Apart from the original messages to be gathered $s_1,\ldots,s_k$, each node $v$ adds to a transmitted message the values of $c(v)$ and $c(\parent(v))$.
%
Note that each node is able to distinguish messages received from its children from other messages 
thanks to the fact that $c(v)\neq c(w)$ if $w\neq v$ is the parent of 
\tj{a neighbor}
$u$ of $v$.
Moreover, $v$ is able to determine when it receives the last message from its children thanks to the fact that it stores the number of its children in the label and it can distinguish messages from different children, as $c(u)\neq c(u')$ for the children $u\neq u'$ of $v$.

\noindent An inductive proof wrt to the levels of $T'$ leads to 
Lemma~\ref{lem:DDelta} (see Appendix~\ref{sec:lem:DDelta}).

\begin{lemma}\label{lem:DDelta}
	The algorithm $A$ is a $k$-gathering algorithm which works in time $O(D \Delta)$ with labels of size $O(\log \Delta)$.
\end{lemma}
\newcommand{\prooflemDDelta}{
\begin{proof}
Let $D'$ be the largest level of $T'$. We will show by induction that each node on the level $l\le D'$ receives messages from all its children until the phase $(D'-l)\Delta$ of an execution of the algorithm $A$. 
For the base case consider nodes on the level $l=D'$. They are leaves of $T'$ thus they do not have children and therefore one can say that they received all messages from their children until the phase $(D'-l)\Delta=$ $0$.
Now, let us make an inductive assumption that all nodes on the levels larger than $l$ for some $l<D'$ received messages from their children until the phase $(D-(l+1))=$ $(D-l-1)\Delta$.
Let $v$ be any node on the level $l$.
Thus, each child $u$ of $v$ transmits its unique message (containing all source messages from its subtree) in such phase $t'$ that $t'\text{ mod } \Delta=$ $s(u)$ and $t'<(D-l-1)\Delta+\Delta=$ $(D-l)\Delta$.
Recall that the properties (a) and (b) of the assignment $s\,:\,V\to [0,\Delta-1]$ guarantee that $s(u)\neq s(w)$ for each other node $w$ on the level $l+1$ which is a neighbor of $v$.
Thus, the message transmitted by $u$ in the phase $t'$ is the unique message transmitted in that phase by a neighbor of $v$ from the level $l+1$. 
Moreover, the other possible neighbors of $v$ belong to the levels $l$ and $l-1$. As the nodes from the levels $l+1,$ $l$ and $l-1$ transmit is separate rounds of the phase $t'$, the message from $u$ is successfully delivered to $v$. This observation concludes the inductive step of the proof.
Thus finally, the $k$-gathering task is accomplished in at most $D\Delta$ phases which gives $O(D\Delta)$ rounds.\qed
\end{proof}
}

\subsection{Summary 
and separation from the model without advice}\label{subsec:k:gather:sum}

Theorem~\ref{t:k-gather} and Lemma~\ref{lem:DDelta}
lead to the following result.
\begin{corollary}
    There exists a distributed algorithm for $k$-gathering in time \\$O\left(\min(D+k,D\Delta)\right)$ 
    with 
    a labeling scheme of optimal length $O(\min(\log\Delta,\log k))$.
\end{corollary}
On the other hand, we prove (Appendix~\ref{proofcorseparation}) the following result.
\begin{corollary}\label{cor:separation}
    The distributed deterministic $k$-gathering problem with arbitrary distinct IDs of nodes requires $\Omega(D\log n + k\log(n/k))$ rounds.
\end{corollary}
\newcommand{\proofcorseparation}{
\begin{proof}
Bruschi et al.\ \cite{BruschiP97} showed that the deterministic broadcasting problem with unique labels requires $\Omega(D\log n)$ rounds. In the proof they build a family of graphs $C_D^n$ for fixed $D<n/2$ such that:
\begin{itemize}
    \item each graph from $C_D^n$ is a layered graph, i.e., the set of nodes $V$ is split into subsets $V_0, V_1,\ldots,V_D$ such that $V_0=\{s\}$, each edge $(u,v)$ connects nodes from consecutive layers, i.e., $u\in V_i$, $v\in V_{i+1}$ for some $i<D$,
    \item for each $u,v\in V_i$ the set of neighbors of $u$ in $V_{i-1}$ and the set of neighbors of $v$ in $V_{i-1}$ are equal.
\end{itemize}
Then, they show that broadcasting requires $\Omega(D\log n)$ rounds for this family of graphs. Observe that the above properties imply that, in each algorithm 
\begin{itemize}
    \item the broadcast message originally given only to $s$ is delivered to all the nodes from $V_i$ at the same round, by the same node from $V_{i-1}$,
    \item the broadcast message can be delivered to any node from $V_i$ only if it was delivered to all nodes from $\bigcup_{j=0}^{i-1} V_j$ earlier
\end{itemize}
for each graph from $C_D^n$. So let us consider the $1$-gathering problem with the source node $s_1$ equal to the only node from $V_0$ and the sink node equal to an arbitrary node from $V_D$. Then, a solution of our instance of the $1$-gathering problem requires broadcasting of the message originally available by $s_1$ only. And using Theorem~3.1 from \cite{BruschiP97}, we obtain the lower bound on $1$-gathering $\Omega(D\log n)$. 

Now, consider a star graph with $k$ sources $s_1,\ldots,s_k$ connected with the sink node $s$. 
If the IDs of nodes might be chosen as arbitrary distinct values from $[n]$ then time needed for delivery of even one of the messages located in the nodes $s_{1},\ldots,s_{k}$ to $s$ corresponds to the size of the $(n,k)$-selective family defined in \cite{ClementiMS03}\footnote{The $(n,k)$-selective family is a family $\mathcal{F}$ of subsets of $[n]$ such that for every non empty subset $Z$ of $[n]$ such that $|Z|\le k$, there is a set $F$ in $\mathcal{F}$ such that $|Z\cap F|=1$; see Def. 2.1 in \cite{ClementiMS03}}. As shown in Theorem 3.3 in \cite{ClementiMS03}, the size of such a family is $\Omega(k \log(n/k) )$.  

We can combine the lower bounds from the previous two paragraphs in the following way. First consider a star graph with the sources $s_1,\ldots,s_k$ and the sink node $s$. Then let the sink $s$ be the source $s\in V_0$ of a graph from the family $C_D^{n-k}$.
The lower bounds from the previous two paragraphs imply the bound $\Omega(D\log n + k\log(n/k))$.
\qed
\end{proof}
}

\noindent As our $k$-gathering algorithm with optimal advice works in $O(D+k)$ rounds,
we obtain a $\Omega(\log n)$ gap between the algorithms with optimal 
advice and the ones 
with arbitrary distinct IDs for each $k=O(n^\eps)$, where $0<\eps<1$ is a constant.

\section{The $k$-broadcasting problems}\label{sec:k:broadcast}

A simple composition of a $k$-gathering algorithm and a single-message broadcasting algorithm lead to the following corollaries (see Appendix~\ref{sec:proof:cor:unit:broadcast} and \ref{sec:proof:cor:multi:broadcast}).
\newcommand{\proofcorunitbroadcast}{
\begin{proof}   
Consider an instance of the unit-message broadcasting problem with the source nodes $s_1,\ldots, s_k$. Let us conceptually add a new node $s$ to the communication graph 
and the edges $\{(s,s_i)\,|\, i\in[1,k]\}$. Then, let us consider an instance of the standard broadcasting problem (i.e., with one source node) with the source node $s$.
The algorithm for broadcasting with advice of size $O(1)$ and time $O(D+\log^2 n)$ from \cite{DBLP:journals/corr/abs-2410-07382} assumes that the source node sends the broadcast message in the first round and it remains inactive afterwards. Thus, we can simulate such an algorithm in our instance of the unit-message $k$-broadcasting such that the nodes just assume that they received the broadcast message in the first round from the (non-existing, in reality) source node $s$ and continue the standard broadcasting algorithm in consecutive rounds.
\qed
\end{proof}
}
\begin{corollary}\label{cor:unit:broadcast}
    There exists a unit-message $k$-broadcasting algorithm with advice size $O(1)$ which works in $O(D+\log^2 n)$ rounds.
\end{corollary}
\newcommand{\proofcormultibroadcast}{
\begin{proof}   
Now, consider 
the multi-message broadcasting problem with the source nodes $s_1,\ldots, s_k$. Consider the following solution. First, we execute a $k$-gathering algorithm with an arbitrary sink vertex $s\in V$ and the sources $s_1,\ldots,s_k$. (The sink node $s$ is chosen by an oracle which builds the labels for nodes of the graph.)
This algorithm takes $O(D+k)$ rounds and requires $O(\min(\log k,\log\Delta))$-bit labels.
Just after reception of all messages from its children, the node $s$ starts an execution of the standard broadcasting algorithm  with advice size $O(1)$ and time $O(D+\log^2 n)$ from \cite{DBLP:journals/corr/abs-2410-07382}. This leads to the statement of the corollary.
\qed\end{proof}}
\begin{corollary}\label{cor:multi:broadcast}
    There exists a multi-message $k$-broadcasting algorithm with the  size of the advice $O(\min(\log k, \log\Delta))$ which works in $O(D+\log^2n+k)$ rounds.
\end{corollary}
\tj{Let us also note here that the $\Omega(D+\log^2n)$ lower bound holds even for centralized unit-message broadcasting.}



\newcommand{\openproblems}{
\section{Open Problems}\label{sec:summary}

An interesting further research direction is to determine optimal length of labeling schemes for $k$-gathering and multi-message $k$-broadcasting in the model without possibility of combining many source messages together. 
}
\openproblems

\vspace*{4pt}
\noindent\textbf{Acknowledgments}
{Supported by the Polish National Science Centre grant 2020/39/B/ST6/03288.}

 \bibliographystyle{splncs04}

\bibliography{references}

\remove{

}

\newpage

\appendix

\noindent{\Large\textbf{Appendix}}

\section{An illustration of the graph $G_{D,p}$}\label{sec:fig:kgathering:lower}
\figkgatheringlower

\section{A figure illustrating an application of Lemma~\ref{l:level} in the $k$-gathering $O(D+k)$-round algorithm} \label{sec:app:figure}

\figparschedulestage

\section{Proof of Lemma~\ref{lem:DDelta}}\label{sec:lem:DDelta}

\prooflemDDelta

\section{Proof of Corollary~\ref{cor:separation}}\label{proofcorseparation}

\proofcorseparation

\section{Proof of Corollary~\ref{cor:unit:broadcast}}\label{sec:proof:cor:unit:broadcast}

\proofcorunitbroadcast

\section{Proof of Corollary~\ref{cor:multi:broadcast}}\label{sec:proof:cor:multi:broadcast}

\proofcormultibroadcast

\end{document}